\documentclass[11pt,reqno]{article}
\usepackage[numbers]{natbib}
\usepackage{amsfonts, amsmath, wasysym}
\usepackage{amssymb, amsthm}
\usepackage{mathrsfs}
\usepackage{verbatim}
\usepackage[usenames,dvipsnames]{color}
\newfont{\ffont}{cmr10}
\usepackage{amsmath, amssymb, amsfonts, amsbsy, amsthm,
latexsym,graphicx}
\usepackage{cancel}
\usepackage{pifont}
\usepackage{enumerate}
\topmargin by -\headsep \textheight 8.9in \oddsidemargin 0pt
\evensidemargin \oddsidemargin \marginparwidth 0.5in

\allowdisplaybreaks

\newcommand{\cH}{\mathcal{H}}
\newcommand{\cA}{\mathcal{A}}

\newcommand{\cX}{\mathcal{X}}

\newcommand{\cN}{\mathcal{N}}

\newcommand{\ot}{\otimes}

\newcommand{\Tr}{\mathrm{Tr}}
\newcommand{\al}{\alpha}

\newcommand{\ga}{\gamma}
\newcommand{\del}{\delta}
\newcommand{\ka}{\kappa}

\newcommand{\Del}{\Delta}

\newcommand{\R}{\mathbb{R}}

\newcommand{\N}{\mathbb{N}}

\newcommand{\E}{\mathbb{E}}

\newcommand{\bP}{\mathbb{P}}

\newcommand{\la}{\lambda}
\newcommand{\eps}{\varepsilon}
\newcommand{\ti}{\times}

\newcommand{\rank}{\mathrm{Rank}}

\newcommand{\Prob}{\mathbb{P}}
\newcommand{\G}{{\mathbb{G}}}

\newcommand{\vare}{{\varepsilon}}

\newcommand{\pack}{\mathrm{pack}}
\newcommand{\card}{\mathrm{card}}

\newcommand{\supp}{\mathrm{supp}}
\newcommand{\dime}{\mathrm{dim}}

\newcommand{\kernel}{\mathrm{ker}}
\newcommand{\vol}{\mathrm{vol}}

\newcommand{\bE}{{\mathbf E}}
\newcommand{\be}{{\mathbf e}}
\newcommand{\bQ}{{\mathbf Q}}
\newcommand{\bx}{{\mathbf x}}
\newcommand{\bv}{{\mathbf v}}

\newcommand{\by}{{\mathbf y}}
\newcommand{\bz}{{\mathbf z}}

\newcommand{\AI}{{\bf A_I}}

\makeatletter

\newcommand{\Rmnum}[1]{\expandafter\@slowromancap\romannumeral #1@}
\makeatother

\makeatletter
\newcommand{\specificthanks}[1]{\@fnsymbol{#1}}
\makeatother

\newcommand{\xspace}{\hbox{\kern-2.5pt}}

\begin{document}

\newtheorem{theorem}{Theorem}[section]
\newtheorem{lemma}[theorem]{Lemma}
\newtheorem{corollary}[theorem]{Corollary}
\newtheorem{proposition}[theorem]{Proposition}

\newcommand{\HR}[1]{{\color{Magenta}{#1}}}
\newcommand{\todo}[1]{{\color{red}{#1}}}

\theoremstyle{definition}
\newtheorem{definition}[theorem]{Definition}
\newtheorem{xca}[theorem]{Exercise}
\newtheorem{design-crit}[theorem]{Design Criterion}
\newtheorem{remark}[theorem]{Remark}

\theoremstyle{remark}
\newtheorem{example}[theorem]{Example}

\title{Uniform recovery of fusion frame structured sparse signals}

\author{Ula\c{s} Ayaz\thanks{Hausdorff Center for Mathematics,
    University of Bonn, Endenicher Allee 60, 53115 Bonn, Germany.} \textsuperscript{,\specificthanks{2}} \and \hspace{0.45in} Sjoerd Dirksen \textsuperscript{\specificthanks{1}} \and
         \and Holger Rauhut\thanks{RWTH Aachen University, Templergraben 55, 52056 Aachen, Germany.}}

\maketitle

\begin{abstract}
We consider the problem of recovering fusion frame sparse signals from incomplete measurements. These signals are composed of a small number of nonzero blocks taken from a family of subspaces. First, we show that, by using a-priori knowledge of a coherence parameter associated with the angles between the subspaces, one can uniformly recover fusion frame sparse signals with a significantly reduced number of vector-valued (sub-)Gaussian measurements via mixed $\ell^1$/$\ell^2$-minimization. We prove this by establishing an appropriate version of the restricted isometry property. Our result complements previous nonuniform recovery results in this context, and provides stronger stability guarantees for noisy measurements and approximately sparse signals. Second, we determine the minimal number of scalar-valued measurements needed to uniformly recover all fusion frame sparse signals via mixed $\ell^1$/$\ell^2$-minimization. This bound is achieved by scalar-valued subgaussian measurements. In particular, our 
result shows that the number of scalar-valued subgaussian measurements cannot be further reduced using knowledge of the coherence parameter. As a special case it implies that the best known uniform recovery result for block sparse signals using subgaussian measurements is optimal.
\end{abstract}

\section{Introduction}

Compressive sensing \cite{carota06,do06-2,Eldar12,FoR13} predicts that sparse signals can be recovered from incomplete and possibly noisy measurements via efficient algorithms. The linear measurement process is typically described via random matrices. For instance, Gaussian random matrices provide optimal recovery guarantees in the sense that $m \geq Cs \log(N/s)$ measurements are necessary and sufficient to recover any $s$-sparse vector in dimension $N$ via $\ell_1$-minimization and other recovery algorithms \cite{cata06,mepato09}.

Often signals possess more structure than just plain sparsity. In the block sparsity model \cite{Eldar08,EKB10} one assumes that a signal consists of blocks, of which only a few are nonzero. This model is strongly related to (and can in fact be viewed as special case of) the joint
sparsity model \cite{elra10,fora08,grrascva08}, where one considers a signal consisting of several ``channels'' (such as the three color channels of an RGB image) and assumes that nonzeros coefficients appear at the same location within each of the channels. A generalization of the block sparsity model is the group sparsity model where the groups of nonzero coefficients are allowed to overlap \cite{BBC13,Rao12}.

In \cite{Boufounos09}, a refinement of the block sparsity model was introduced which is related to the concept of fusion frames. In the fusion frame sparsity model we assume that the signal is block sparse and, in addition, lies in
$$
\cH=\{\bx = (x_j)_{j=1}^N : x_j \in W_j, \ \forall j \in [N]\} \subset \R^{dN},
$$
i.e., every nonzero block $x_j$ of the signal is assumed to lie in a certain $k$-dimensional subspace $W_j$ of $\R^d$. The collection of these subspaces may form a fusion frame (although this is not strictly required for our theory). Fusion frames generalize frames \cite{ch03} and were first introduced in \cite{Casazza04} under the name of `frames of subspaces' (see also the survey
\cite{Casazza13}). They allow to analyze signals by projecting them onto multidimensional subspaces and for stable reconstruction from these projections.

In this paper we study uniform recovery of fusion frame sparse signals using random linear measurements. We are particularly interested in determining whether one can reduce the number of measurements needed for recovery if one knows that the subspaces are \emph{incoherent}. We measure the coherence with the parameter
$$\lambda = \max_{i \neq j} \|P_iP_j\|_{2\to 2},$$
where $P_i$ is the orthogonal projection onto the subspace $W_i$. This parameter, which was introduced in \cite{Boufounos09}, is a measure of the mutual orthogonality of the subspaces (note that $\lambda=0$ means that all subspaces are orthogonal to each other).

In the first part of the paper, we consider a measurement model in which we take scalar-valued linear measurements of the signal, i.e., we observe a vector $\by$ of $m'$ measurements of the form
\begin{equation}
\label{eqn:scalarMeasIntro}
\by = \mathbf{\Phi} \bx
\end{equation}
where $\mathbf{\Phi} \in \R^{m'\ti dN}$. This recovery problem was studied, among others, in \cite{EKB10,elmi09,Stojnic09-1}. As we recall in Section~\ref{sec:sufScalar}, it is essentially known that one can uniformly recover every $s$-sparse signal in $\cH$ in a stable and robust manner from $m'$ subgaussian measurements using mixed $\ell_1/\ell_2$-minimization, provided that $m'\gtrsim s\log(N/s) + sk$. Our first main result, stated in Theorem~\ref{uni_optimal}, shows that this is \emph{optimal}: if one can recover every $s$-sparse vector in $\cH$ from the linear measurements in (\ref{eqn:scalarMeasIntro}) via mixed $\ell_1/\ell_2$-minimization, then $m'\gtrsim s\log(N/s) + sk$ measurements are \emph{necessary}. In particular this shows that the number of scalar subgaussian measurements cannot be further reduced by using a-priori knowledge of the coherence of the subspaces $W_j$.

In the second part of the paper, we consider a more natural measurement model in which one takes vector-valued measurements. In this model we observe $m$ measurements of the form
$$
\by=(y_i)_{i=1}^m = \left(\sum_{j=1}^N a_{ij} x_j \right)_{i=1}^m,
\ \ y_i \in \R^d.
$$
In the vector-valued measurement model, one expects the coherence of the subspaces $W_j$ to play a significant role. To see this, note that if the spaces $W_j$ are mutually orthogonal, then a \emph{single} measurement is sufficient to recover the signal. Indeed, in this case $x_j = a_j^{-1}P_jy$. This suggests that fewer measurements are necessary when the subspaces are close to being orthogonal.

In \cite{Boufounos09} the first results concerning the reconstruction of fusion frame sparse signals from vector-valued measurements via mixed $\ell_1/\ell_2$-minimization were obtained. In particular, a fusion frame version of the well-known restricted isometry property (RIP) \cite{cata06}, \cite[Chapter~6]{FoR13} was introduced and it was shown that is implied by the classical RIP. As a consequence, a signal that is $s$-sparse in this fusion frame model can be recovered from $m \geq C s \log(N/s)$ Gaussian vector-valued measurements, which corresponds to taking the measurement coefficients $a_{ij}$ to be Gaussian. This bound does not take the coherence of the subspaces $W_j$ into account. A different analysis in \cite{Boufounos09} based on the block-coherence of the measurement matrix yields the desired dependence in $\la$, but this type of analysis cannot provide the optimal scaling of the measurements in terms of the sparsity. A decrease of the number of measurements when $\la$ decreases was
also observed in an average case analysis in \cite{Boufounos09}, where the non-zero coefficients are chosen at random according to a certain model, similarly as in \cite{elra10}. Recently in \cite{Ayaz13}, a nonuniform recovery result valid for a fixed fusion frame sparse signal and Bernoulli measurements was shown, which requires fewer measurements as $\la$ decreases and at the same time exhibits linear scaling in terms of the sparsity (up to logarithmic factors). The proof proceeds by an adaptation of the golfing scheme invented by D.~Gross \cite{gr09-2}, see also \cite[Chapter~12]{FoR13}.

The second main result of this paper, Theorem~\ref{main_uniform}, guarantees uniform recovery of sparse vectors via $\ell_1$/$\ell_2$-minimization using a small number of vector-valued (sub)gaussian measurements. Our recovery guarantee exhibits improved stability in the presence of noise and under approximate sparsity when compared to the non-uniform recovery result in \cite{Ayaz13}. Contrary to the earlier results in \cite{Boufounos09}, the required number of measurements simultaneously decreases with $\lambda$ and shows linear scaling in terms of the signal sparsity up to logarithmic factors. We establish this result by showing that a subgaussian matrix satisfies the fusion restricted isometry property with high probability. Our proof of the latter result relies heavily on a recent tail bound for suprema of second order chaos processes \cite{KMR13}. In the final section of the paper we give a lower bound on $\la$ and use this to compare a necessary condition on the number of measurements needed for uniform
sparse recovery via mixed $\ell_1$/$\ell_2$-minimization with the sufficient condition in Theorem~\ref{main_uniform}. These conditions do not match, however, and we leave it as an interesting open problem to close the gap between the two.

\subsection*{Acknowledgements}

All authors would like to thank the Hausdorff Center for Mathematics for support. Ula\c{s} Ayaz and Holger Rauhut acknowledge funding
by the European Research Council (ERC) through the Starting Grant 258926 (SPALORA).

\subsection{Block sparsity and fusion frames}\label{fusionintro}

We consider signals $\bx = (x_j)_{j = 1}^N \in \R^{dN}$, where the components $x_j \in \R^d$ are vectors themselves.
We say that $\bx$ is \emph{$s$-sparse} if it is $s$-sparse in the block sense \cite{Eldar08}, i.e., $\|\bx\|_0\leq s$, where
$$
\|\bx\|_0 = \card \{j \in [N] : x_j \neq 0\}.
$$
Here, and in the following $[N]:= \{1,2,\hdots,N\}$. We refine this block sparsity model in the following way.
Given a collection of $N$ subspaces $W_j \subset \R^d$ with $\text{dim}(W_j)=k$, $j\in [N]$,
we require that the components $x_j$ of a vector $\bx$ are contained in $W_j$ for all $j \in [N]$, that is,
$\bx$ is contained in the space
$$
\cH=\{\bx = (x_j)_{j=1}^N : x_j \in W_j, \ \forall j \in [N]\} \subset \R^{dN}.
$$
Accordingly, for any $s\in [N]$ we denote the $s$-sparse vectors in $\cH$ by
$$\cH_s = \{\bx \in \cH \ : \ \|\bx\|_0\leq s\}.$$
Often, vectors are only approximately sparse. In order to measure how close a given vector is to the set of sparse vectors
we introduce the error of best $s$-term approximation in $\ell_{2,1}$ of a vector $\bx$ as
\begin{equation}\label{best:sterm}
\sigma_s(\bx)_1 := \inf_{\|\bz\|_0 \leq s} \|\bx -\bz \|_{2,1},
\end{equation}
where
\begin{align}
\|\bx\|_{2,1} = \sum_{j=1}^N \|x_j\|_2. \label{l21norm}
\end{align}
We call vectors with small $\sigma_s(\bx)_1$ \emph{approximately sparse} or \emph{compressible}.

In what follows, it will be important to introduce a parameter that measures how much the subspaces $W_j$ deviate
from a collection of orthogonal subspaces. Let $P_j : \R^d \to \R^d$ denote the orthogonal projection onto the subspace $W_j \subset \R^d$.
We define the \emph{coherence} $\la$ of $(W_j)_{j=1}^N$ by
\begin{align}\label{lambda_par}
\lambda = \max_{i \neq j} \|P_iP_j\|_{2\to 2},
\end{align}
where $\|\cdot\|_{2\to 2}$ denotes the operator norm. Let $\theta_{ij}^{(1)}\leq \theta_{ij}^{(2)} \leq \ldots \leq \theta_{ij}^{(k)}$ be the principal angles between $W_i$ and $W_j$. The cosines of the principal angles coincide with the singular values of $P_iP_j$ (see e.g.\ \cite{Ste73}) so that
\begin{align}\label{lambda_alter}
\lambda = \max_{i \neq j} \cos \theta_{ij}^{(1)}.
\end{align}

In order to explain the terminology {\em fusion frame sparsity model}, we recall that the collection $(W_j)_{j=1}^N$ of subspaces is called a {\em fusion frame}
if there are constants $0<A\leq B < \infty$ (called frame bounds) and certain weights $v_j>0$, $j \in [N]$, such that
\begin{align}\label{FF_property}
A \|x\|_2^2 \leq \sum_{j=1}^N v_j^2\|P_j x\|_2^2 \leq B\|x\|_2^2 \quad \mbox{ for all } x \in \R^d.
\end{align}
The special case that all subspaces $W_j$ are one-dimensional, i.e., $W_j = \operatorname{span}\{f_j\}$ and $v_j = \|f_j\|_2$ reduces to the situation of classical frames \cite{ch03}. We refer to \cite{Casazza13} for more information on fusion frames. Although our work, as well as \cite{Ayaz13,Boufounos09}, is strongly motivated by fusion frames,
we emphasize that our results below do not assume that the subspaces $W_j$ satisfy the fusion frame property \eqref{FF_property}.

\section{Sparse recovery using scalar-valued measurements}

In this section we consider the recovery of sparse vectors from scalar-valued linear measurements. That is, we assumed that we are given a vector $\by$ of $m'$ measurements
\begin{equation}
\label{eqn:scalarMeas}
\by = \mathbf{\Phi} \bx
\end{equation}
where $\mathbf{\Phi} \in \R^{m'\ti dN}$. Our goal is to recover a sparse $\bx \in \cH$ from these measurements. This problem can be formulated as the optimization program
\begin{align*}
(L_0) \qquad \hat{\bx}= \text{argmin}_{\bx \in \cH} \|\bx\|_0 \quad
\mbox{subject to } \qquad \mathbf{\Phi}\bx=\by.
\end{align*}
which is NP-hard. Following \cite{Boufounos09,elra10,fora08,tr06-4}, we instead use the convex program
\begin{align*}
(L_{2,1}) \qquad \hat{\bx}= \text{argmin}_{\bx \in \cH} \|\bx\|_{2,1} \quad
\mbox{ subject to } \qquad \mathbf{\Phi}\bx=\by,
\end{align*}
where the $\|\cdot\|_{2,1}$-norm is defined in \eqref{l21norm}. We shall refer to this program as either mixed $\ell_1/\ell_2$-minimization or $\ell_{2,1}$-minimization. This problem can be solved efficiently.\par
Similarly to the situation in classical compressive sensing, there are alternative methods available to recover sparse signals. We mention in particular the `block' versions of matching and orthogonal matching pursuit studied in \cite{EKB10}. In this paper we concentrate exclusively on $\ell_{2,1}$-minimization.

\subsection{A sufficient condition for uniform recovery}
\label{sec:sufScalar}

Let us first review the known uniform recovery results using random scalar measurements. The results are phrased in terms of the following restricted isometry constants on $\cH_s$, the set of $s$-sparse vectors in $\cH$. If $W_j=\R^d$ for all $j \in [N]$, then our definition coincides with the block restricted isometry constants introduced in \cite{elmi09}.
\begin{definition}
\label{def:RIPHs}
The restricted isometry constant of $\mathbf{\Phi}\in \R^{m'\ti dN}$ on $\cH_s$ is the smallest constant $\theta_s\geq 0$ satisfying
$$(1-\theta_s)\|\bx\|_2^2 \leq \|\mathbf{\Phi}\bx\|_2^2 \leq (1+\theta_s)\|\bx\|_2^2, \qquad \mathrm{for  \ all} \  \bx\in\cH_s.$$
\end{definition}
The following result is essentially \cite[Theorem 2]{elmi09}. It also follows from the proof of \cite[Theorem 4.4]{Boufounos09} (simply replace $\AI$ by $\mathbf{\Phi}$ there).
\begin{theorem}\label{thm:robustScalar}
Suppose that the RIP constant $\theta_{2s}$ of $\mathbf{\Phi}$ on $\cH_s$ satisfies $\theta_{2s} < \sqrt{2}-1$. For $\bx \in \cH$, let noisy measurements $\by = \mathbf{\Phi} \bx + \be$ be given with $\|\be\|_2 \leq \eta$. Let $\hat{\bx}$ be the solution of the $\ell_{2,1}$-minimization program
$$\min_{\bz \in \cH} \|\bz\|_{2,1} \quad \mbox{ subject to }\quad \|\mathbf{\Phi}\bz - \by\|_2 \leq \eta.$$
Then
\begin{align*}
\|\bx - \hat{\bx}\|_2 \leq C_1 \frac{\sigma_s(\bx)_1}{\sqrt{s}} +
C_2 \eta,
\end{align*}
and
$$
\|\bx - \hat{\bx}\|_{2,1} \leq C_1 \sigma_s(\bx)_1 +
C_2 \sqrt{s} \eta,
$$
where the constants $C_1, C_2 >0$ only depend on $\theta_{2s}$.
\end{theorem}
Since $\cH_s$ can be viewed as a union of $N\choose s$ subspaces of dimension $sk$, the next result follows from general bounds on restricted isometry constants of subgaussian matrices acting on unions of subspaces. We refer to \cite[Theorem 3]{Blu11} and \cite[Corollary 5.4]{Dir14} for details.
\begin{theorem}
\label{thm:subgaussianScalar}
If $\mathbf{\Phi}$ is an $m'\ti dN$ random matrix with independent, $\al$-subgaussian rows, then $\bP(\theta_s\geq \theta)\leq \eps$ if
\begin{equation}
\label{eqn:numScalMeas}
m'\gtrsim \al^4\theta^{-2}\max\{s\log(eN/s) + sk, \log(\eps^{-1})\}.
\end{equation}
\end{theorem}
As a particular consequence of Theorems~\ref{thm:robustScalar} and \ref{thm:subgaussianScalar}, if $m'$ satisfies (\ref{eqn:numScalMeas}), then with probability at least $1-\eps$ we can robustly recover any $s$-sparse vector $\bx \in \cH$ exactly via $\ell_{2,1}$-minimization from $m'$ scalar-valued subgaussian measurements. One may wonder whether it is possible to decrease this number of measurements if the subspaces $W_j$ are incoherent. In the following section we will show that the answer is negative: one needs at least $m'\gtrsim s\log(eN/s) + sk$ measurements to recover every $s$-sparse vector exactly via $\ell_{2,1}$-minimization. Note that this also implies that the condition on $m'$ in Theorem~\ref{thm:subgaussianScalar} is essentially optimal.

\subsection{A necessary condition for uniform recovery}\label{necessary}

In this section, we investigate the minimal number of scalar-valued measurements required for uniform recovery of sparse signals via $\ell_{2,1}$-minimization. We follow the method used in \cite{Foucart10,FoR13} to establish the minimal number of measurements needed for sparse recovery via $\ell^1$-minimization in classical compressed sensing.
\begin{theorem}\label{uni_optimal}
Let $\mathbf{\Phi} \in \R^{m' \times dN}$ and let $(W_j)_{j=1}^N$ be a collection of subspaces in $\R^d$ with $\dime(W_j)=k$. If every $4s$-sparse vector $\bx \in \cH$ is a minimizer of $\min_{\bz \in \cH} \|\bz\|_{2,1}$ subject to $\mathbf{\Phi}\bz=\mathbf{\Phi}\bx$, then
\begin{align}\label{optimal_cond}
m' \geq c_1 s \log \left(\frac{N}{c_2 s} \right) + c_3sk
\end{align}
where $c_1 \approx 0.46$, $c_2=32$ and $c_3 \approx 0.18$.
\end{theorem}
In the proof of Theorem~\ref{uni_optimal} we use the following combinatorial lemma (see e.g.\ \cite[Lemma 10.12]{FoR13} for a proof).
\begin{lemma}\label{subsets}
Given integers $s < N$, there exists an $n\in \N$ such that
\begin{align}\label{subset_count}
n \geq \left( \frac{N}{8s} \right)^{s}
\end{align}
and subsets $I_1, \ldots, I_n$ of $[N]$ such that each $I_i$ has cardinality $2s$ and
$$ \card(I_i \cap I_\ell) < s \ \ \text{ whenever } i \neq \ell.$$
\end{lemma}

\subsubsection{Covering and packing numbers}

In the proof of Theorem~\ref{uni_optimal} we use some covering number estimates. Recall the following standard terminology. Let $T$ be a subset of a metric space $(X,d)$.
For $t > 0$, the covering number $\cN(T,d,t)$ is defined as the smallest integer $\cN$ such that $T$ can be covered with balls $B(x_\ell,t)= \{x
\in X, d(x,x_\ell) \leq t \}$, $x_\ell \in T$, $\ell \in [\cN]$,
i.e.,
\begin{equation}
\label{eqn:ballCover}
T \subset \bigcup_{\ell =1}^\cN B(x_\ell,t).
\end{equation}
The packing number $\mathcal{P}(T,d,t)$ is defined, for $t > 0$, as the
maximal integer $\mathcal{P}$ such that there are points $x_\ell \in
T$, $\ell \in [\mathcal{P}]$, which are $t$-separated, i.e.,
$d(x_\ell,x_k) > t$ for all $k,\ell \in [\mathcal{P}], k \neq \ell$. If $X= \R^n$ is a normed vector space and the metric $d$ is induced
by a norm $\|\cdot\|$ via $d(u,v)= \|u-v\|$, then we also write $\cN(T,\|\cdot\|,t)$ and $\mathcal{P}(T,\|\cdot\|,t)$.
Next we state a variation of a well-known covering number bound, see e.g.\ \cite[Lemma 4.16]{pi99}.
\begin{lemma}\label{covering_number}
Let $\|\cdot\|$ be any norm on $\R^n$ and let $U=\{x\in \R^n:
\frac{1}{2} \leq \|x\| \leq 2\}$. For any $t > 0$,
\begin{align}
\left(\frac{2}{t}\right)^n - \left(\frac{1}{2t}\right)^n \leq \cN(U, \|\cdot\|,t) \leq \mathcal{P}(U, \|\cdot\|,t).
\end{align}
\end{lemma}
\begin{proof}
For the first inequality, let $\{x_1,\ldots, x_\cN\} \subset U$ be a minimal set satisfying (\ref{eqn:ballCover}). Let $B$ be the unit ball with respect to $\|\cdot\|$ and let $\vol$ be the Euclidean volume on $\R^n$. Then
$$
\cN \vol(tB) \geq \vol \left( \bigcup_{\ell=1}^\cN B(x_\ell,
  t) \right) \geq \vol(U) = \vol(2B) - \vol(B/2).
$$
Since $\vol(tB)=t^n\vol(B)$, we have $\cN t^n \vol(B) \geq 2^n \vol(B) - (1/2)^n\vol(B)$. This yields $\cN \geq (2/t)^n - (1/2t)^n$ as desired. The second inequality is standard, see e.g.\ \cite[Lemma C.2]{FoR13}.
\end{proof}
We will also use the following standard estimate, see e.g.\ \cite[Proposition C.3]{FoR13} for a proof.
\begin{lemma}\label{covering_number2}
Let $\|\cdot\|$ be any norm on $\R^n$ and let $U$ be a subset of the unit ball $B=\{x\in \R^n: \|x\| \leq 1\}$. Then for any $t > 0$,
\begin{align}
\cN(U, \|\cdot\|,t) \leq \mathcal{P}(U, \|\cdot\|,t) \leq \left(1+\frac{2}{t}\right)^n.
\end{align}
\end{lemma}
Moreover, we will need the \textit{Gilbert-Varshamov} bound from coding theory \cite{Gilbert52,Varshamov57}.
\begin{lemma}\label{Gilbert}(Gilbert-Varshamov bound) Let $A_1,\ldots,A_l$ be sets, each consisting of $q$ elements. Let $d_H$ be the Hamming distance on $A_1\ti\cdots\ti A_l$. Then, for any $t\in [l]$,
\begin{align*}
\mathcal{P}(A_1\ti\cdots\ti A_l,d_H,t) \geq \frac{q^l}{\sum_{i=0}^{t-1} \binom{l}{i} (q-1)^i }.
\end{align*}
\end{lemma}

\subsubsection{Proof of Theorem~\ref{uni_optimal}}

First recall that we restrict our attention to the set
$$
\cH=\{\bx=(x_i)_{i=1}^N : x_i \in W_i, \ \forall i \in [N]\}
\subset \R^{dN}.
$$
Let $\mathbf{\Phi}_{|\cH}$ be the restriction of $\mathbf{\Phi}$ to $\cH$ and consider the quotient space
$$ X:= \cH \ / \ \kernel \ \mathbf{\Phi}_{|\cH} = \{ [\bx]:= \bx + \kernel \ \mathbf{\Phi}_{|\cH}, \ \bx \in \cH\},$$
which is equipped with the quotient norm
$$ \| [\bx]\| := \inf_{\bv \in \kernel \ \mathbf{\Phi}_{|\cH}} \|\bx- \bv\|_{2,1}, \
\bx \in \cH.
$$
Let $B_X$ be the unit ball of $X$
with respect to this norm. Given a $4s$-sparse vector $\bx \in \cH$, we notice that every vector $\bz=\bx-\bv$ with $\bv \in \kernel
\ \mathbf{\Phi}_{|\cH}$ satisfies $\mathbf{\Phi} \bz = \mathbf{\Phi} \bx$. Thus, the
assumption of the theorem
gives $\|[\bx]\| = \|\bx\|_{2,1}$. \\
Next we pack the spherical shells $S_i = \{ y \in W_i : \frac{1}{2s}
\ \leq \|y\|_2 \leq \frac{2}{s} \}$ of the subspaces $W_i$.
Since $\dim(W_i)=k$, Lemma~\ref{covering_number} yields that
we can find such a packing with packing distance of $1/s$ where
\begin{align}\label{pack_lower}
\mathcal{P}(S_i,\|\cdot\|_2,1/s) \geq 2^k -(1/2)^k =: q.
\end{align}
For each $i$, fix a packing set $T_i$ of $S_i$ with this cardinality.
Now let $I_1,\ldots, I_n$ be the sets introduced in Lemma~\ref{subsets}, and for each $j \in [n]$ define the set
of $2s$-sparse vectors with support set from $I_j$
$$
\mathcal{T}(I_j) := \{ \bx \in \cH : \supp (\bx) = I_j, \ x_i
\in T_i \ \text{ for } i \in I_j\}.
$$
It is clear that for each $j$, the total number of vectors in $\mathcal{T}(I_j)$ is at least $q^{2s}$. We wish to find a large
enough subset of each $\mathcal{T}(I_j)$, say $\widehat{\mathcal{T}}(I_j)$, such that any two distinct vectors $\bx,\by \in
\widehat{\mathcal{T}}(I_j)$ satisfy $\card\{i: x_i \neq y_i\} \geq s$.
Lemma~\ref{Gilbert}, applied with $l=2s$ and $t=s$, says that there exists such a subset $\widehat{\mathcal{T}}(I_j)$ satisfying
\begin{align*}
\card(\widehat{\mathcal{T}}(I_j)) &\geq \frac{q^{2s}}{\sum_{i=0}^{s-1}
\binom{2s}{i} (q-1)^i } \geq \frac{q^{2s}}{2^{2s} q^s} \geq (q/4)^s.
\end{align*}
Let us define the set
$$
\mathcal{V}= \bigcup_{j=1}^n \widehat{\mathcal{T}}(I_j).
$$
We claim that the set $\{[\bx] : \bx \in \mathcal{V}\}$ is a 1-separated subset of the scaled unit ball $4 B_X$.
Observe that a vector $\bx \in \mathcal{V}$ is supported on some $I_j$, $j \in [n]$, of cardinality $2s$. Therefore,
$$
\|[\bx]\| = \|\bx\|_{2,1} = \sum_{i \in I_j} \|x_i\|_2 \leq 2s \cdot
\frac{2}{s} = 4,
$$
since $x_i \in S_i$. We now distinguish between two cases. First
assume that $\bx, \by \in \mathcal{V}$ both belong to the same set
$\widehat{\mathcal{T}}(I_j)$ for some $j$. Then we know that $\card\{i:
x_i \neq y_i\} \geq s$ holds. Whenever $x_i \neq y_i$, we have
$\|x_i-y_i\|_2 \geq 1/s$ since $x_i, y_i \in T_i$. Hence,
\begin{align*}
\|[\bx] - [\by] \| = \|[ \bx - \by]\| = \|\bx -\by \|_{2,1} =
\sum_{i: x_i \neq y_i} \|x_i-y_i\|_2 \geq s \cdot \frac{1}{s} = 1.
\end{align*}
The second equality holds since $\bx - \by$ is $4s$-sparse. For the
case that $\bx \in \widehat{\mathcal{T}}(I_j)$ and $\by \in
\widehat{\mathcal{T}}(I_\ell)$ with $j \neq \ell$, observe that
$\bx-\by$ is $4s$-sparse and the symmetric difference of the support
sets satisfies $\card(I_j \bigtriangleup I_\ell) \geq 2s$. It
follows that
\begin{align*}
\|[\bx] - [\by] \| &= \|[ \bx - \by]\| = \|\bx -\by \|_{2,1} \geq
\sum_{i\in I_j \bigtriangleup I_\ell} (\|x_i\|_2 + \|y_i\|_2) \geq 2s
\cdot \frac{1}{2s} = 1,
\end{align*}
since $x_i,y_i \in S_i$. This proves our claim that $\mathcal{V}$
separates $4 B_X$ by 1 and $\card(\mathcal{V}) \geq n (q/4)^s$. The space $X$ has dimension $r := \rank(\mathbf{\Phi}) \leq m'$. Applying Lemma~\ref{covering_number2} with $4B_X$, this implies that $n (q/4)^s\leq 9^r \leq 9^{m'}$. In view of \eqref{subset_count}, we obtain for $k \geq 1$
$$
 9^{m'} \geq \left(\dfrac{N}{8s}\right)^s \left(\frac{2^k
    -(1/2)^k}{4}\right)^s \geq \left(\dfrac{N}{32s}\right)^s \left(\frac{3}{2}\right)^{k s},
$$
since $2^k - (1/2)^k \geq (3/2)^k$ for $k \geq 1$.  Taking the
logarithm on both sides gives the desired result.
\qed

\section{Sparse recovery using vector-valued measurements}

We now consider the problem of recovering a fusion frame sparse vector $\bx=(x_j)_{j=1}^N \in \cH$ from vector-valued measurements. Our model consists of taking $m$ linear combinations, i.e.,
$$
\by=(y_i)_{i=1}^m = \left(\sum_{j=1}^N a_{ij} x_j \right)_{i=1}^m,
\ \ y_i \in \R^d.
$$
The coefficient matrix $A=(a_{ij})_{i\in [m],j\in [N]} \in \R^{m \times N}$ is called the measurement matrix in this context. From a mathematical perspective, we can view the vector-valued measurement scheme as a special case of (\ref{eqn:scalarMeas}) as follows. To the matrix $A$ we associate the Kronecker product $\AI:= A\ot I_d$, which can be represented as a block matrix ${\bf A_I}=(a_{ij} I_d)_{i \in [m], j\in[N]} \in \R^{dm \times dN}$, where $I_d$ is the identity matrix of size $d \times d$. With this notation we can concisely formulate our measurement scheme as
\begin{equation}
\label{eqn:vectorMeas}
\by = {\bf A_I} \bx.
\end{equation}
Even though (\ref{eqn:vectorMeas}) is a special case of (\ref{eqn:scalarMeas}), we remark that it is appropriate to consider $m$ (rather than $m'=md$) as the number of measurements, since this represents the number of performed measurement operations.\par
Our goal is to recover a sparse $\bx \in \cH$ from these measurements. Analogously to the case of scalar-valued measurements, we do this using the $\ell_{2,1}$-minimization program
\begin{align*}
(L_{2,1}) \qquad \hat{\bx}= \text{argmin}_{\bx \in \cH} \|\bx\|_{2,1} \quad
\mbox{ subject to } \qquad \AI \bx=\by,
\end{align*}
where the $\|\cdot\|_{2,1}$-norm is defined in \eqref{l21norm}.

\subsection{A nonuniform recovery result}

Before giving details on the main uniform recovery result of this paper, we recall a result from \cite{Ayaz13} where the first and third named author considered the recovery of a fixed sparse signal from random measurements.

Recall that a random variable $\xi$ is called $\al$-subgaussian if $\Prob(|\xi| >t) \leq 2 \exp(-t^2/2\al^2)$. The entries $A_{ij} = \xi_{ij}$ of an $m \times N$ \emph{$\al$-subgaussian matrix} $A$ are independent, mean-zero, variance one, $\al$-subgaussian random variables. Examples of $1$-subgaussian matrices are the standard Gaussian matrix, whose entries are independent standard Gaussian random variables, and the Bernoulli matrix, whose entries are independent random variables taking the values $\pm 1$ with equal probability.
\begin{theorem}\label{nonuni_main}Let $(W_j)_{j=1}^N \subset \R^d$ be $k$-dimensional subspaces with coherence
$\lambda \in [0,1]$ and fix $\bx \in \cH$. Let $A \in
\R^{m \times N}$ be a Bernoulli or Gaussian matrix and assume that
\begin{align}\label{nonuni_cond}
m \geq C (1 + \lambda s) \log^{\beta}( Nsk )
\log(\vare^{-1}).
\end{align}
Let noisy measurements $\by = \AI \bx + \be$
be given with $\|\be\|_2 \leq \eta \sqrt{m}$. Let $\hat{\bx}$ be the solution
of the convex optimization problem
\begin{equation*}
\min_{\bx \in \cH} \|\bx\|_{2,1} \quad \mbox{ subject to }\quad \|\AI \bx - \by\|_2 \leq \eta \sqrt{m}.
\end{equation*}
Then with probability at least $1-\vare$,
\begin{equation}\label{nonuni:error}
\|\bx -\hat{\bx}\|_2 \leq C_1 \sigma_s(\bx)_1 + C_2 \sqrt{s} \eta,
\end{equation}
where we recall that $\sigma_s(\bx)_1$ is the error of best $s$-term approximation \eqref{best:sterm}.
The constants $C,C_1,C_2 >0$ are universal. Here
$\beta=1$ in the Bernoulli case and $\beta=2$ in the Gaussian
case.
\end{theorem}
The condition $\|\be\|_2 \leq \eta \sqrt{m}$ is
natural for a vector $\be = (e_j)_{j=1}^m$. For instance, it is implied
by the bound $\|e_j\|_2 \leq \eta$ for all $j \in [m]$.
Theorem~\ref{nonuni_main} implies exact sparse recovery via the equality constrained $\ell_{2,1}$-minimization problem $(L_{2,1})$
when $\bx$ is $s$-sparse and $\eta=0$. The required sufficient number $m$ of samples in
\eqref{nonuni_cond} decreases essentially linearly with $\la$. Furthermore, the exponent $\beta=2$ in the Gaussian case is likely not optimal, but presently there is no better estimate available.

\subsection{Main result}
Our uniform recovery result using subgaussian vector-valued measurements reads as follows.
\begin{theorem}\label{main_uniform}
Let $A \in \R^{m \times N}$ be an $\al$-subgaussian matrix and let $(W_j)_{j=1}^N \subset \R^d$
be $k$-dimensional subspaces with coherence $\lambda \in [0,1]$. Assume that
\begin{align}\label{uniform_cond}
m\geq C \al^4 \max\{(\log^2(s)+\la s)(k + \log(N)),\log(\eps^{-1})\}.
\end{align}
Then with probability at least $1-\vare$, $(L_{2,1})$ recovers \textit{all} $s$-sparse $\bx$ from
$\by=\AI \bx$. Moreover, with probability at least $1-\vare$, every vector $\bx \in \cH$ is approximated by
a minimizer $\hat{\bx}$ of
$$\min_{\bz \in \cH} \|\bz\|_{2,1} \quad \mbox{ subject to }\quad \left\|\frac{1}{\sqrt{m}}\AI\bz - \by\right\|_2 \leq \eta$$
with $\by=\frac{1}{\sqrt{m}}\AI \bx + \be$ and $\|\be\|_2\leq \eta$ in the sense that
$$
\|\bx - \hat{\bx}\|_2 \leq C_1 \frac{\sigma_s(\bx)_1}{\sqrt{s}} +
C_2 \eta \ \ \ \text{and} \ \ \
\|\bx - \hat{\bx}\|_{2,1} \leq C_1 \sigma_s(\bx)_1 +
C_2 \sqrt{s} \eta,
$$
where the constants $C_1, C_2 >0$ are universal.
\end{theorem}
\begin{remark}
In Section~\ref{sec:compareNecSuf}, we show that the term $\lambda s$ in
\eqref{uniform_cond} is almost optimal under special conditions on the orientation of the subspaces.
However, we believe that the linear factor $k$ in \eqref{uniform_cond} is suboptimal. It is shown in Remark~\ref{subK} that it is not possible to remove this factor with the methods we use to prove our result. As it is now, this factor can be
ignored in the range $k \lesssim \log(N)$. Also note that in the special case of classical frames, $k = 1$, so that the factor of $k$ in the number of measurements vanishes.\par
The assumption that the dimensions of the subspaces $W_j$ are all equal is not necessary in Theorem~\ref{main_uniform}, and can be relaxed to $\max_{1\leq j\leq N}\mathrm{dim}(W_j)\leq k$ in the general case. 		
\end{remark}
Let us briefly compare Theorem~\ref{main_uniform} to earlier uniform recovery results in the literature. In the conference paper \cite{AyazConf13-2}, the first and third named authors announced a preliminary version of Theorem~\ref{main_uniform}. They followed a similar approach as pursued in the present paper but could only provide the condition
$$
m \gtrsim k \sqrt{\lambda s^2 + s}  \log^4( Nd )
$$
on the sufficient number of measurements for uniform recovery with subgaussian matrices. Theorem~\ref{main_uniform} improves on the dependence in $\lambda$, $k$ and $s$ and on the number of logarithmic factors. As we noted earlier, it was shown in \cite{Boufounos09} that $m \gtrsim s \log(N/s)$ measurements are sufficient for recovery of fusion frame sparse signals using
many random measurement ensembles. This result is obtained by showing that if the underlying measurement matrix $A$ satisfies the
classical RIP, then for an arbitrary collection of subspaces $(W_j)_{j=1}^N$ in $\R^d$, the associated matrix $\AI$ satisfies the FRIP, see Remark~\ref{RIPimp}. The paper \cite{Boufounos09} also analyzes recovery via an adapted version of the coherence of the measurement matrix $A \in \R^{m \times N}$ having $\ell_2$-normalized columns $a_1,\hdots,a_N$, defined as $\mu_f := \max_{j \neq k} |\langle a_j, a_k \rangle| \| P_j P_k\|$. A fusion frame $s$-sparse signal $\bx$ can be recovered exactly via $\ell_{2,1}$-minimization if $s < (1+\mu_f^{-1})/2$ \cite[Theorem 3.5]{Boufounos09}. Clearly, $\mu_f \leq \lambda \mu$, where $\mu$ is the standard version of the coherence, $\mu :=  \max_{j \neq k} |\langle a_j, a_k \rangle|$. For matrices $A$ with near-optimal coherence of the order $\mu \sim c/\sqrt{m}$, we conclude that fusion frame $s$-sparse signals can be recovered from $m \geq C \lambda^2 s^2$ measurements. This bound shows a better scaling in $\lambda$ but a quadratic, and hence not optimal, scaling in $s$.

\section{Proof of Theorem~\ref{main_uniform}}

To prove Theorem~\ref{main_uniform} we follow a by now classical approach \cite{cata06,FoR13} to analyze sparse recovery for various algorithms via the restricted isometry property (RIP). The following version of the RIP, which is appropriate in the context of the vector-valued measurement model (\ref{eqn:vectorMeas}), was introduced in \cite{Boufounos09}.
\begin{definition}
Let $A \in \R^{m \times N}$ and $(W_j)_{j=1}^N$ be a collection of
subspaces in $\R^d$. The fusion restricted isometry constant $\delta_s$ of $A$ is the smallest constant such that
\begin{align}\label{FRIP}
(1-\delta_s) \|\bx\|_2^2 \leq \|\AI \bx\|_2^2 \leq
(1+\delta_s)\|\bx\|_2^2, \qquad \mathrm{for \ all} \ \bx\in \cH_s.
\end{align}
We loosely say that $A$ satisfies the fusion restricted isometry property (FRIP) if it has small fusion restricted isometry constants.
\end{definition}
Notice that $\del_s(A)$ is exactly equal to $\theta_s(\AI)$, which was introduced in Definition~\ref{def:RIPHs}. In this section we prove Theorem~\ref{thm:FRIP} which gives a sufficient condition under which an (appropriately scaled) subgaussian matrix satisfies the FRIP. In combination with Theorem~\ref{thm:robustScalar} (applied with $\mathbf{\Phi}=\frac{1}{\sqrt{m}}\AI$), this implies Theorem~\ref{main_uniform}.\par
We first collect some ingredients for the proof of Theorem~\ref{thm:FRIP}. Given an $m \ti n$ matrix $A$ we let $\|A\|_{2\to 2}$ and $\|A\|_F$ denote its operator and Frobenius norm, respectively. If $\cA$ is a set of $m\ti n$ matrices,
then we define the radii of $\cA$ in these norms by
$$d_{2\to 2}(\cA) = \sup_{A \in \cA} \|A\|_{2\to 2}, \qquad d_{F}(\cA) = \sup_{A \in \cA} \|A\|_{F}.$$
We let $\ga_2(\cA,\|\cdot\|_{2\to 2})$ denote the $\ga_2$-functional
of $\cA$. We do not give a precise definition of this quantity (see \cite{Dir13,KMR13} for more information),
but recall that it can be estimated by an entropy integral
(see e.g.\ \cite[Section 1.2]{Tal05})
\begin{equation}
\label{eqn:gammaFunEstEntInt}
\ga_2(\cA,\|\cdot\|_{2\to 2}) \lesssim \int_0^{d_{2\to 2}(\cA)} \log^{1/2}(\cN(\cA,\|\cdot\|_{2\to 2},u)) \ du.
\end{equation}
To estimate the right hand side, we will use below that for any $c,\ka>0$ (see e.g.\ \cite[Lemma C.9]{FoR13}),
\begin{equation}
\label{eqn:intEstimate}
\int_0^{\ka} \log^{1/2}(c/u) \ du \leq \ka\log^{1/2}\Big(\frac{ec}{\ka}\Big).
\end{equation}
The most important ingredient for our proof is the following tail
bound for suprema of second order chaos processes from \cite{KMR13}.
We state here a sharpened version from \cite[Theorem 6.5]{Dir13}. To this end, we introduce the $\psi_2$-norm of a random variable $\xi$
as
\[
\|\xi\|_{\psi_2} := \inf \{C > 0: \E[\exp(\xi^2/(2C^2))] \leq 2\}.
\]
A random variable $\xi$ is subgaussian if and only if $\|\xi\|_{\psi_2}$ is finite.

\begin{theorem}
\label{thm:chaos}
Let $\cA$ be a set of $m\ti n$ matrices. Suppose that
$\xi_1,\ldots,\xi_n$ are independent, real-valued, mean-zero random
variables, let $\xi=(\xi_1,\ldots,\xi_n)$ and set
$\|\xi\|_{\psi_2} = \max_i\|\xi_i\|_{\psi_2}$.
Then there are constants $c,C>0$ such that for any $u\geq 1$,
\begin{align}
\label{eqn:supChaosTail}
\bP\Big(\sup_{A\in\cA}\Big|\|A\xi\|_2^2 - \E\|A\xi\|_2^2\Big| &
\geq C\|\xi\|_{\psi_2}^2\Big(\ga_2^2(\cA,\|\cdot\|_{2\to 2}) + d_F(\cA)\ga_2(\cA,\|\cdot\|_{2\to 2})\Big) \nonumber \\
& \ \ \ \ \ + c\|\xi\|_{\psi_2}^2\Big(\sqrt{u}d_{F}(\cA)d_{2\to 2}(\cA) + ud_{2\to 2}^2(\cA)\Big)\Big) \leq e^{-u}.
\end{align}
\end{theorem}
In the proof of Theorem~\ref{thm:FRIP} we use the following special case of the dual Sudakov inequality (\cite{PTJ86}, see also \cite[Section 3.3]{LeT91}).
\begin{lemma}
\label{lem:dualSudakovNorm}
Let $\|\cdot\|$ be any norm on $\R^n$. If $g$ denotes an $n$-dimensional standard Gaussian vector, then
$$\sup_{u>0} u\log^{1/2}(\cN(B_2^n,\|\cdot\|,u)) \leq \E\|g\|.$$
\end{lemma}
We are now prepared to present our result.
\begin{theorem}
\label{thm:FRIP}
Let $(W_j)_{j=1}^N\subset \R^d$ be $k$-dimensional subspaces with coherence $\la$. Let $A$ be an $m\ti N$ $\al$-subgaussian matrix and let $\del_s$ be the fusion restricted isometry constant of $\frac{1}{\sqrt{m}}\AI$.
Then $\bP(\del_s\geq\del)\leq \eps$ provided that
\begin{equation}
\label{eqn:measurementsFRIP}
m\geq C \al^4\del^{-2}\max\{(\log^2(s)+\la s)(k + \log(N)),\log(\eps^{-1})\}.
\end{equation}
\end{theorem}
\begin{proof}
From the definition of the fusion restricted isometry constant we have
$$
\delta_s = \sup_{\bx \in D_{s,N}} \left|\frac{1}{m}\|\AI \bx \|_2^2 - \|\bx\|_2^2 \right| = \frac{1}{m} \sup_{\bx
\in D_{s,N}} \left|  \|\AI \bx \|_2^2 - \E\|\AI \bx\|_2^2 \right|,
$$
where
$$
D_{s,N} = \{\bx=(x_j)_{j=1}^N \in \cH \ : \ \|\bx\|_0\leq s, \ \|\bx\|_2\leq 1\}.
$$
The relation $ \E\|\AI \bx\|_2^2 =
m \|\bx\|_2^2$ follows easily by noticing that the entries $\xi_{ij}$ of $A$ are independent random variables and all have mean zero and variance 1. Now observe that
\begin{align*}
\AI\bx = \sum_{i \in [m], j \in [N]} \xi_{ij}
\bQ_{ij}\bx,
\end{align*}
where $\bQ_{ij} := \bE_{ij}(I_d)$
denotes the $m\ti N$ block matrix with block $(i,j)$ equal to $I_d$ and all other blocks zero.
We define the matrix $V_\bx \in \R^{md \times mN}$ whose columns
are $ \bQ_{ij}\bx$ for $i \in [m], j \in [N]$, i.e.,
$$
V_\bx = (\bQ_{11}\bx | \bQ_{12}\bx | \ldots |
\bQ_{mN}\bx) \in \R^{md \times mN}.
$$
Then $\AI \bx= V_\bx \boldsymbol\xi$ where $\boldsymbol\xi=(\xi_{ij})$ is a
subgaussian vector of length $m N$ and we can write
$$
\delta_s = \frac{1}{m} \sup_{\bx \in D_{s,N}} \left|\|V_\bx
\boldsymbol\xi\|_2^2 - \E \|V_\bx \boldsymbol\xi\|_2^2 \right|.
$$
We will deduce the result by applying (\ref{eqn:supChaosTail}) to the supremum on the right hand side.
For this purpose we first calculate $\ga_2(\{V_\bx \ : \ \bx\in D_{s,N}\},\|\cdot\|_{2\to 2})$ and the radii of $\{V_\bx \ : \ \bx \in D_{s,N}\}$ in the operator and Frobenius norms.
Consider $\bx,\by\in D_{s,N}$ and set $\bz=\bx-\by$, so that $V_\bz = V_\bx-V_\by$. Observe that $V_\bz^* V_\bz$ is a block diagonal matrix with the same $N\ti N$ block
$(\langle z_{j},z_{\ell}\rangle)_{j \in [N],\ell\in [N]}$ repeating $m$ times. This first of all shows that
$$
\|V_\bz\|_F^2 = \Tr(V_\bz^*V_\bz) = m\sum_{j=1}^N \|z_j\|_2^2
$$
and in particular the radius in the Frobenius norm satisfies $d_F(\{V_\bx \ : \ \bx\in D_{s,N}\}) = \sqrt{m}$. Moreover,
$$
\|V_\bz\|_{2\to 2}^2 = \|V_\bz^*V_\bz\|_{2\to 2} = \|(\langle z_{j},z_{\ell}\rangle)_{j\in [N],\ell\in [N]}\|_{2\to 2}=:\|B_\bz\|_{2\to 2}.
$$
We split $B_\bz$ in two parts. Let $D_\bz$ be the $N\ti N$ diagonal matrix with $\|z_j\|_2^2$ as diagonal elements and set $C_\bz=B_\bz-D_\bz$. By the triangle inequality, we obtain
\begin{align*}
\|(\langle z_{j},z_{\ell}\rangle)_{j \in [N],\ell\in [N]}\|_{2\to 2} & \leq \|D_\bz\|_{2\to 2} + \|C_\bz\|_{2\to 2} \leq \|D_\bz\|_{2\to 2} + \|C_\bz\|_F \\
& = \max_{j\in [N]} \|z_j\|_2^2 + \Big(\sum_{j\neq \ell}\langle z_j,z_{\ell}\rangle^2\Big)^{1/2} \\
& \leq \max_{j\in [N]} \|z_j\|_2^2 + \Big(\sum_{j\neq \ell}\la^2 \|z_j\|_2^2\|z_{\ell}\|_2^2\Big)^{1/2} \\
& \leq \max_{j\in [N]} \|z_j\|_2^2 + \la\sum_{j=1}^N \|z_j\|_2^2.
\end{align*}
This shows that the radius in the operator norm satisfies $d_{2\to 2}(\{V_\bx : \bx\in D_{s,N}\}) \leq \sqrt{1+\la}$.
Moreover, if we set $\|\bz\|_{2,\infty} = \max_{j\in [N]} \|z_j\|_2$ then
$$\|V_\bz\|_{2\to 2} \leq \|\bz\|_{2,\infty} + \sqrt{\la} \|\bz\|_2.$$
Combining this with (\ref{eqn:gammaFunEstEntInt}),
\begin{align}
\label{eqn:EIsumNorm}
\ga_2(\{V_\bx \ : \ \bx\in D_{s,N}\},\|\cdot\|_{2\to 2}) & \lesssim \int_0^{\infty} \log^{1/2}(\cN(\{V_\bx \ : \ \bx\in D_{s,N}\},\|\cdot\|_{2\to 2},u)) \ du \nonumber\\
& \leq \int_0^{\infty}\log^{1/2}(\cN(D_{s,N},\|\cdot\|_{2,\infty} + \sqrt{\la}\|\cdot\|_2, u)) \ du.
\end{align}
Let $\Del = \{(\bx,\bx) \ : \ \bx\in D_{s,N}\}\subset D_{s,N}\times D_{s,N}$ and define
$$\|(\bx,\by)\|_* = \|\bx\|_{2,\infty} + \la \|\by\|_2 \qquad ((\bx,\by) \in D_{s,N}\ti D_{s,N}).$$
Then,
$$\cN(D_{s,N},\|\cdot\|_{2,\infty} + \sqrt{\la}\|\cdot\|_2,u) = \cN(\Del,\|\cdot\|_*,u) \leq \cN(D_{s,N}\ti D_{s,N},\|\cdot\|_*,u).$$
Now, if $N_1$ is a $u$-net for $D_{s,N}$ in the $\|\cdot\|_{2,\infty}$-norm and $N_2$ is a $u$-net for $D_{s,N}$ in the $\sqrt{\la}\|\cdot\|_2$-norm,
then $N_1\ti N_2$ is a $2u$-net for $D_{s,N}\ti D_{s,N}$ in the $\|\cdot\|_*$-norm. Therefore,
$$\cN(D_{s,N}\ti D_{s,N},\|\cdot\|_*,u) \leq \cN(D_{s,N},\|\cdot\|_{2,\infty},u/2)\cN(D_{s,N},\sqrt{\la}\|\cdot\|_{2},u/2).$$
Combining these observations with (\ref{eqn:EIsumNorm}), we obtain
\begin{align}
\label{eqn:gamma2Split}
& \ga_2(\{V_\bx \ : \ \bx\in D_{s,N}\},\|\cdot\|_{2\to 2}) \nonumber \\
& \qquad \lesssim \int_0^{\infty}\log^{1/2}\Big(\cN(D_{s,N},\|\cdot\|_{2,\infty},u/2)\cN(D_{s,N},\sqrt{\la}\|\cdot\|_{2},u/2)\Big) \ du \nonumber \\
& \qquad \leq 2\int_0^{1}\log^{1/2}(\cN(D_{s,N},\|\cdot\|_{2,\infty},u)) \ du \nonumber \\
& \qquad \qquad \qquad \qquad + 2\sqrt{\la}\int_0^{1}\log^{1/2}(\cN(D_{s,N},\|\cdot\|_{2},u)) \ du,
\end{align}
where in the final estimate we used that the $\|\cdot\|_{2,\infty}$ and $\|\cdot\|_2$-diameters of $D_{s,N}$ are equal to $1$. The second integral can be estimated via the volumetric estimate in Lemma~\ref{covering_number2}. It implies that
\begin{align}
\label{eqn:covEstElem}
\cN(D_{s,N},\|\cdot\|_2,u) & \leq \sum_{S\subset [N],|S|=s} \cN(\{\bx \in D_{s,N} \ : \ \mathrm{supp}(\bx)=S\},\|\cdot\|_2,u) \nonumber\\
& \leq \Big(\frac{eN}{s}\Big)^s\Big(1+\frac{2}{u}\Big)^{sk},
\end{align}
where we used that the dimension of $\{\bx \in D_{s,N} \ : \ \mathrm{supp}(\bx)=S\}$ is $sk$, and find using (\ref{eqn:intEstimate})
$$
\int_0^{1}\log^{1/2}(\cN(D_{s,N},\|\cdot\|_{2},u)) \ du
\lesssim s^{1/2}\log^{1/2}(eN/s) + (sk)^{1/2}.
$$
We now compute the first entropy integral on the right-hand side of (\ref{eqn:gamma2Split}).
Since we are now considering the $\|\cdot\|_{2,\infty}$ instead of the $\|\cdot\|_2$-norm we expect that $s$
does not play a role in this part. We use the dual Sudakov inequality to verify this.
For every $1\leq j\leq N$ we let $U_j:\R^k\rightarrow W_j\subset \R^d$ be the matrix with columns consisting
of an orthonormal basis of $W_j$. Define $\mathbf{U}:\R^{kN}\rightarrow \cH$ by $\mathbf{U}=(U_1,\ldots,U_N)$, then $\mathbf{U}^*\mathbf{U}=I_{kN}$.
We consider the norm on $\R^{kN}$ given by
$$
\|\bx\|_{2,\infty,\mathbf{U}} = \|\mathbf{U}\bx\|_{2,\infty}
$$
Let $\mathbf{g}$ be a $kN$-dimensional standard Gaussian vector. By Lemma~\ref{lem:dualSudakovNorm},
\begin{align*}
\sup_{u>0} u\log^{1/2}(\cN(D_{s,N},\|\cdot\|_{2,\infty},u)) & \leq \sup_{u>0} u\log^{1/2}(\cN(B_2^{kN},\|\cdot\|_{2,\infty,\mathbf{U}},u)) \\
& \lesssim \E\|\mathbf{g}\|_{2,\infty,\mathbf{U}} = \E\|\mathbf{U}\mathbf{g}\|_{2,\infty}.
\end{align*}
Since $U_j^*U_j = I_k$ for all $1\leq j\leq N$, we see that
\begin{align*}
\E\|\mathbf{U}\mathbf{g}\|_{2,\infty} & = \E\max_{1\leq j\leq N} \|U_jg_j\|_2 \\
& \leq \max_{1\leq j\leq N} \E\|g_j\|_2 + \E\max_{1\leq j\leq N}\Big|\|g_j\|_2 - \E\|g_j\|_2\Big|,
\end{align*}
where $g_j$ denotes the $j$-th $k$-dimensional block of $\mathbf{g}$. Clearly, for all $1\leq j\leq N$, $\E\|g_j\|_2\leq(\E\|g_j\|_2^2)^{1/2}=k^{1/2}$. Moreover, observe that the function $x\mapsto\|x\|_2$ is $1$-Lipschitz.
Therefore, the concentration inequality for Lipschitz functions of Gaussian vectors (see e.g.\ \cite[Theorem 5.5]{BLM13}) implies that
$$\Big|\|g_j\|_2 - \E\|g_j\|_2\Big|$$
is $1$-subgaussian and therefore
$$\E\max_{1\leq j\leq N}\Big|\|g_j\|_2 - \E\|g_j\|_2\Big| \lesssim \log^{1/2} N.$$
Using these estimates we find for any $0<\ka\leq 1$,
\begin{align*}
\int_{\ka}^1 \log^{1/2}(\cN(D_{s,N},\|\cdot\|_{2,\infty},u)) \ du & \lesssim (k^{1/2}+\log^{1/2} N) \int_{\ka}^1u^{-1} \ du \\
& = (k^{1/2}+\log^{1/2} N)\log(\ka^{-1}).
\end{align*}
By our earlier calculation (\ref{eqn:covEstElem}) and (\ref{eqn:intEstimate}),
\begin{align*}
& \int_0^{\ka} \log^{1/2}(\cN(D_{s,N},\|\cdot\|_{2,\infty},u)) \ du \\
& \qquad \leq \int_0^{\ka} \log^{1/2}(\cN(D_{s,N},\|\cdot\|_2,u)) \ du \\
& \qquad \leq \int_0^{\ka} s^{1/2}\log^{1/2}(eN/s) + (sk)^{1/2}\log^{1/2}\Big(1+\frac{2}{u}\Big) \ du \\
& \qquad \leq \ka s^{1/2}\log^{1/2}(eN/s) + (sk)^{1/2}\ka\log^{1/2}(3e\ka^{-1}).
\end{align*}
Combining these two estimates for $\ka=s^{-1/2}$, we find
\begin{align*}
& \int_0^1 \log^{1/2}(\cN(D_{s,N},\|\cdot\|_{2,\infty},u)) \\
& \qquad \lesssim \tfrac{1}{2}(k^{1/2}+\log^{1/2} N)\log(s) + \log^{1/2}(eN/s) + k^{1/2}\log^{1/2}(3es^{1/2}).
\end{align*}
Collecting our estimates we arrive at
\begin{align*}
& \ga_2(\{V_\bx \ : \ \bx\in D_{s,N}\},\|\cdot\|_{2 \to 2}) \\
& \qquad \lesssim (\la s)^{1/2}\log^{1/2}(eN/s) + (\la sk)^{1/2} + (k^{1/2}+\log^{1/2} N)\log(s) \\
& \qquad \qquad + \log^{1/2}(eN/s) + k^{1/2}\log^{1/2}(s).
\end{align*}
By applying (\ref{eqn:supChaosTail}) with these estimates shows that $\del_s\leq \del$ with probability $1-\eps$ if
$$m\geq C \al^4\del^{-2}\max\{(\log^2(s)+\la s)(k + \log(N)),\log(\eps^{-1})\}.$$
\end{proof}
\begin{remark}\label{subK}
Let us compare the lower bound on the number of required measurements
in (\ref{eqn:measurementsFRIP}) to the requirement for nonuniform recovery stated in Theorem~\ref{nonuni_main}.
Note that $m$ scales in the same way in $\la$, $s$ and $N$ as in
(\ref{nonuni_cond}) (and better in $\eps$). However, the scaling in $k$ is now linear instead of logarithmic.
We conjecture that it is possible to remove the dependence on $k$
altogether. We note here, however, that it is not possible to deduce a positive answer using Theorem~\ref{thm:chaos}.
Indeed, the diagonal projection is a contraction for the operator norm
and therefore, using the notation from the proof of Theorem~\ref{thm:FRIP}, we have for any $1\leq j_*\leq N$,
$$\|z_{j_*}\|_2^2 \leq \max_{j\in [m]} \|z_j\|_2^2 = \|D_\bz\|_{2\to 2} \leq \|B_\bz\|_{2\to 2} = \|V_\bz\|_{2\to 2}^2.$$
Since the unit ball of $W_{j_*}$ is embedded in $D_{s,N}$, we conclude by the majorizing measures theorem \cite{Tal05}
and by identifying $W_{j_*}$ with $\R^k$ that
\begin{align*}
\ga_2(\{V_\bx \ : \ \bx\in D_{s,N}\},\|\cdot\|_{2\to 2}) & \geq \ga_2(B_2^k,\|\cdot\|_2) \gtrsim \E\sup_{x\in B_2^k}\langle x,g\rangle = \E\|g\|_2\gtrsim \sqrt{k}.
\end{align*}
In other words, the linear scaling of $m$ in $k$ in
(\ref{eqn:measurementsFRIP}) is an inevitable consequence of the application of Theorem~\ref{thm:chaos} (although the just deduced lower bound
does not exclude that a possible refinement of our method leads to a bound where $k$ only appears as an additive term). Note that if the dimensions of the subspaces $W_j$ are not equal, then this argument shows that
$$\ga_2(\{V_\bx \ : \ \bx\in D_{s,N}\},\|\cdot\|_{2\to 2})\gtrsim \max_{1\leq j\leq N} \mathrm{dim}(W_j).$$
\end{remark}
\begin{remark}\label{RIPimp}
One can show that the fusion restricted isometry constant $\del_s$ is always bounded by the classical restricted isometry constant $\tilde{\del}_s$ \cite[Proposition 4.3]{Boufounos09}. In particular, in the context of Theorem~\ref{thm:FRIP} this implies that $\bP(\del_s\geq\del)\leq \eps$ if
\begin{equation*}
m\geq C \al^4\del^{-2}\max\{s\log(N/s),\log(\eps^{-1})\},
\end{equation*}
since under this condition $\bP\left(\tilde{\del}_s(\frac{1}{\sqrt{m}}A)\geq \del \right)\leq\eps$ (see e.g.\ \cite[Theorem 9.2]{FoR13}).
This condition is better than the result in Theorem~\ref{thm:FRIP} if the coherence $\la$ is close to $1$.
\end{remark}

\section{A lower bound on the coherence parameter}\label{packing}

In the most favorable scenario, where the coherence $\la$ is (close to) zero, Theorem~\ref{thm:FRIP}
implies that with high probability one can recover any $s$-sparse signal in a stable and robust manner
using ($L_{2,1}$) and a number of measurements which scales only \emph{logarithmically} in the sparsity $s$ and the number of subspaces $N$. This scenario occurs, for example, if $W_1,\ldots,W_N$ are (nearly)
orthogonal lines in $\R^d$ (and $N\leq d$). However, if $N>d/k$ then one can no longer pick $N$ orthogonal
subspaces $k$-dimensional subspaces in $\R^d$, so that $\la=0$ is impossible. In this section we investigate
what the best result is that can be extracted from Theorem~\ref{thm:FRIP} by establishing a lower bound on
$\la$ in terms of the dimensional parameters $k,d$ and $N$.\par
To derive a quantitative lower bound on $\la$ we relate it to optimal packings of Grassmannian manifolds. Let $\G(k,\R^d)$ denote the real Grassmannian
manifold, the collection of all $k$-dimensional subspaces of $\R^d$. We consider the following metric on $\G(k,\R^d)$. Given $V,W \in \G(k,\R^d)$ and their principle angles
$\theta^{(1)}\leq\theta^{(2)}\leq\ldots\leq\theta^{(k)}$, the \emph{spectral distance} between $V$ and $W$ is given by
$$
d_s (V,W) := \min_{\ell} \sin \theta^{(\ell)} = \sin \theta^{(1)}.
$$
The metric $d_s$ is directly related to the coherence parameter $\la$. Indeed, by (\ref{lambda_alter}),
\begin{align}\label{angleREL}
\lambda^2 = \max_{i \neq j} \left( \cos \theta_{ij}^{(1)}
\right)^2 = \max_{i \neq j} \left( 1 - \left( \sin \theta_{ij}^{(1)}
\right)^2 \right) = 1 - \min_{i \neq j} d_s^2(W_i, W_j).
\end{align}
Thus, finding a sharp lower bound for $\la$ is equivalent to finding a set $\cX=(W_j)_{j=1}^N$ in $\G(k,\R^d)$ with maximal packing diameter
$$
\pack_{d_s}(\mathcal{X}) := \min_{i \neq j} d_s(W_i,W_j).
$$
The following upper bound for the packing diameter appears in \cite[Corollary 4.2]{Dhillon08}. Theorem~\ref{thm:specBound} is a direct consequence of a packing diameter bound for the chordal distance \cite[Corollary 5.2]{Conway96}.
Recall that a subspace packing is called \textit{equi-isoclinic} if all the principal angles between all pairs of subspaces are identical \cite{Lemmens73}.
\begin{theorem}
\label{thm:specBound}
If $\cX$ is a set of $N$ subspaces in $\G(k,\R^d)$, then
\begin{align}\label{specbound}
\pack_{d_s}^2(\mathcal{X}) \leq  \frac{(d-k)}{d} \frac{N}{N-1}.
\end{align}
If the bound is attained, then $\cX$ is equi-isoclinic.
\end{theorem}
Optimal packings do not always exist, and for some parameter choices of $N,d,k$ it is not known whether they exist.
It is shown in \cite[Theorem 3.5]{Lemmens73}
that the maximum number of equi-isoclinic $k$-dimensional subspaces in $\R^d$ cannot be greater than
\begin{align}\label{isoclinic}
\frac{1}{2}d(d+1) -\frac{1}{2} k (k+1) +1.
\end{align}
Let us now state a lower bound on the parameter $\la$. By combining \eqref{angleREL} and \eqref{specbound} we obtain
\begin{align}
\label{lambdabound}
\lambda^2 \geq 1- \frac{(d-k)}{d} \frac{N}{N-1}  = \frac{k N-d}{d
N-d}.
\end{align}
When $k=1$, this bound gives a lower bound on the coherence $\mu$ of $N$ $\ell_2$-normalized vectors in $\R^d$, which is known as the \textit{Welch bound}.

\subsection{Comparison with a necessary condition for sparse recovery}
\label{sec:compareNecSuf}

Taking $\mathbf{\Phi}=\AI$ and $m'=md$ in Theorem~\ref{uni_optimal} shows that one needs at least
\begin{equation}
\label{eqn:necVector}
m\gtrsim \frac{s}{d}\log(N/s) + \frac{sk}{d}
\end{equation}
vector-valued measurements to recover every $4s$-sparse vector in $\cH$ exactly using ($L_{2,1}$). Since (\ref{eqn:necVector}) does not involve the coherence parameter $\lambda$, it is not immediately clear how to compare this result with the sufficient conditions for sparse recovery in Theorems~\ref{nonuni_main} and \ref{thm:FRIP}. By (\ref{lambdabound}),
$$\la^2\geq \frac{kN-d}{dN-d}.$$
Thus, in the most favorable case (assuming $N$ is large), $\la\sim\sqrt{k/d}$. In this situation, according to Theorem~\ref{main_uniform}
$$m\gtrsim \Big(\log^2(s)+s\sqrt{\frac{k}{d}}\Big)(k + \log(N))$$
is sufficient for uniform recovery with a Bernoulli matrix. Note that there is a gap between this condition and the necessary condition (\ref{eqn:necVector}). We leave it as an interesting open problem to close this gap.

\end{document}